\documentclass[preprint,12pt,authoryear]{elsarticle}

\usepackage{xcolor}
\usepackage{subcaption}
\usepackage{float}

\usepackage{amsmath}
\usepackage{amsthm}
\usepackage{amssymb}
\usepackage{fcsymbol}
\usepackage{amsfonts}
\usepackage{nccmath}
\usepackage{dsfont}
\newtheorem{proposition}{Proposition}
\newtheorem{assump}{Assumption}

\newcommand\asim{\mathrel{\overset{\makebox[0pt]{\mbox{\scriptsize a}}}{\sim}}}

\newcommand*\e{\mathbb{E}}

\DeclareMathOperator*{\plim}{plim}

\expandafter\def\expandafter\normalsize\expandafter{%
    \normalsize%
    \setlength\abovedisplayskip{8pt}%
    \setlength\belowdisplayskip{8pt}%
    \setlength\abovedisplayshortskip{8pt}%
    \setlength\belowdisplayshortskip{8pt}%
}

\usepackage{tikz}
\usepackage[nomessages]{fp}
\usetikzlibrary{decorations.pathreplacing,arrows,backgrounds,positioning,fit,matrix}

\usepackage{algorithm,algorithmic}

\usepackage{multirow}
\usepackage{booktabs}

\usepackage[hidelinks]{hyperref} %
    \hypersetup{colorlinks, citecolor=, linkcolor=, urlcolor=blue}

\journal{arXiv}

\begin{document}

\begin{frontmatter}



\title{Modelling with Sensitive Variables} 


\author[chan]{Felix Chan}

\affiliation[chan]{organization={School of Accounting, Economics and Finance, Curtin University}}

\author[matyas]{L\'aszl\'o M\'aty\'as}

\affiliation[matyas]{organization={Department of Economics, Central European University}}

\author[reguly1,reguly2]{\'Agoston Reguly\corref{cor1}}
\cortext[cor1]{Corresponding author.\\ 
\textit{Acknowledgments}: Contribution by Balázs Kertész to earlier versions of this paper is kindly acknowledged. We are grateful for all the useful comments on different versions of this paper by William Green, Tom Wansbeek, Botond Köszegi, Robert Lieli, Gábor Békés, Miklós Koren, Michael Knaus, participants at the Central European University's Seminar, the IAAE and EEA conferences. Special thanks to the International Association for Applied Econometrics for the financial support of the 2019 conference presentation.\\
Replication codes are available at \href{https://github.com/regulyagoston/Split-sampling}{https://github.com/regulyagoston/Split-sampling} }
\affiliation[reguly1]{
organization={Institute of Economics, Corvinus University of Budapest}}
\affiliation[reguly2]{organization={Scheller College of Business, Georgia Institute of Technology}}

\begin{abstract}
The paper deals with models in which the dependent variable, some explanatory variables, or both represent sensitive data. We introduce a novel discretization method that preserves data privacy when working with such variables. A multiple discretization method is proposed that utilizes information from the different discretization schemes. We show convergence in distribution for the unobserved variable and derive the asymptotic properties of the OLS estimator for linear models. Monte Carlo simulation experiments presented support our theoretical findings. Finally, we contrast our method with a differential privacy method to estimate the Australian gender wage gap.
\end{abstract}




\end{frontmatter}



\nocite{reguly2024doc_app}
\section{Introduction}
Over the last decade, there has been an explosion in the amount of data collected on individuals, firms, the economy, and society. Governments and the private sector gather large amounts of information that researchers and analysts use on a daily basis. Privacy constraints to share these data or willingness to respond in case of surveys, often prevent access to sensitive data in its original form. We introduce here a novel sampling method that naturally protects individual privacy or increases the willingness to answer through discretization, while minimizes information loss from an estimation perspective.
%
\par
Income is perhaps one of the most commonly used sensitive variable in economics. Instead of providing the sensitive data itself or asking directly for the actual weekly personal income, we work with discrete categories such as \textit{`weekly personal income below $\$100$, between $\$100$ and $\$400$, or above $\$400$}. It is well known in the econometric literature that this simple discretization leads to modeling problems (see, e.g., \citealt{hsiao1983regression}, \citealt{manski2002inference}, \citealt{beresteanu2008asymptotic}, \citealt{beresteanu2011discretized}, \citealt{pacini2019interval}, or \citealt{abrevaya2020interval_fe}), which we solve by using multiple discretization schemes. In the context of income, instead of using only one discretization scheme, we use multiple discretizations by changing the interval boundaries. For example, discretization scheme (2): below $\$50$, between $\$50$ and $\$350$ or above $\$350$; discretization scheme (3): below $\$150$, between $\$150$ and $\$450$ or above $\$450$, etc. We call the method \textit{split sampling},\footnote{The term \textit{split sampling} in this paper is not related to the technique occasionally used in chromatography (see e.g., \citealp{split1977gas1}) or cross-validation methods in machine learning, which splits the initial sample into folds.} and show convergence in distribution for the discretized sensitive variable to the original unobserved variable as both the number of observations and the number of discretization schemes go to infinity.
%
%
\par
Our solution is simple in the sense that each discretization scheme allows (maintains) data protection or enhances the willingness to respond, while the properties of the underlying distribution are recovered by the combination of the resulting discretized variable. Building on our result of convergence in the distribution, we investigate three types of regression: i) the sensitive variable is an explanatory variable; ii) the sensitive variable is the outcome variable; and iii) sensitive variables are on both sides of the regression model. We use our split sampling method to point identify parameters in such linear models and derive the properties of the OLS estimator. The solution to identifying parameters in such context, is to condition the linear regression model on the appropriate known discretization intervals. The definition of the intervals, varies based on where the sensitive variable(s) are in the regression. Then split sampling allows to approximate the unknown distribution of this variable and obtain the conditional expectations needed to identify the parameter of interest without revealing their actual values. 
We show that this general method yields a consistent and asymptotically normal estimator for the parameters of interest in linear models. We then extend our discussion to nonlinear models and panel data as well. To demonstrate our method's finite sample properties, we provide some Monte Carlo evidence and apply our procedure to wage data from the Australian Tax Office (ATO) to investigate the gender wage gap in Australia.
%
%
\par
Our method complements the large literature on differential privacy (DP) to protect data privacy. DP has received much attention in the last decade and has become the industry standard.\footnote{It has been used by many technology companies, including Google \cite{erlingsson2014rappor}, Microsoft \citep{ding2017collecting}, LinkedIn \citep{kenthapadi2018pripearl}, or Facebook. The United States Census Bureau employed DP in 2020 to safeguard individual confidentiality in the U.S. decennial census.} DP quantifies the notion of privacy for downstream tasks and aims to protect the most extreme observations as well. (see, e.g., \citealt{dwork2006differentialprivacy}, \citealt{jordan2015machine}, or \citealt{bi2023distribution}). There are two large branches of differential privacy; the first is ``standard'' or ``central'' DP, where data owners can only publish randomized statistics. Our paper relates to the second branch, which is called ``local'' DP, where data owners do not trust the ``central server'' or are not allowed/want to share the data in its original form. Local DP has multiple variants\footnote{\cite{dwork2014algorithmic} uses a Laplace mechanism, \cite{rohde2020geometrizing} uses local differential privacy, or \cite{avella2021privacy} proposes parametric estimation, just to list a few variants.} and the major issue is that such data privatization may alter the analysis (estimation) results vis-a-vis the original data. As this problem is not new, \cite{Duchi02012018} and \cite{cai2021costofdp} work out this trade-off between privacy and statistical accuracy in different setups, such as the mean estimator or least-squares parameters in linear regression.
In the econometric context, \cite{bi2023distribution} proposes a DP discretization that maintains statistical accuracy for a pre-specified model, while allowing for data privacy. Our approach to some extent is close to theirs in the sense of maintaining (asymptotic) consistency; however, this is the only link it shares with the differential privacy literature. For example, instead of adding noise (Laplace or any other) to the data, we employ discretization to ``distort'' the information.\footnote{Compared to the novel method introduced by \cite{bi2023distribution}, in which they require specifying the model before the analysis, our discretization process is model-agnostic in the sense that it can be used with different model specifications. Furthermore, \cite{bi2023distribution} achieves asymptotic convergence through splitting the sample into two parts and providing only one part to the end user. The other part is then used to transform the original data, add noise, and convert back to a representative distribution of the original data, resulting in an efficiency loss in the number of observations. Our method does not require such sample splitting.}
%
%
\par
Our approach also relates to the literature on partially identified parameters (e.g., see \citealt{manski2003partial}, \citealt{manski2002inference}, \citealt{tamer2010partial}, \citealt{molinari2020microeconometrics}, \citealt{abrevaya2020interval_fe}, \citealt{pacini2019interval}, and \citealt{wang2022partial}). This literature starts from the fact that the variable of interest is discretized (also called interval censored) for some reason (e.g., \citealp{pacini2019interval} discuss income surveys). \cite{manski2002inference} show that the conditional expectation function and the regression parameters in such context cannot be point identified in general. They show how to derive set identification when the discretized variable is on the right-hand side. \cite{beresteanu2008asymptotic}, \cite{bontemps2012set} discuss the case when the outcome variable is discretized, while \cite{beresteanu2011discretized} cover the case when both outcome and covariate(s) are discretized. The main toolkit of partially identified parameters in such context is moment (in)equalities.\footnote {Moment (in)equality models generalize the problems to multiple equations and/or inequalities, see for examples, \cite{chernozhukov2007estimation} or \cite{andrews2010inference}. \cite{beresteanu2008asymptotic} show the asymptotic properties of such partially identified parameters. \cite{imbens2004confidence}, \cite{chernozhukov2007estimation}, and \cite{kaido2019confidence}, among others, derive confidence intervals for these set-identified parameters. These methods are feasible ways to estimate parameter sets for a given conditional expectation function without any further assumptions.} The main advantage of these methods is that they do not require adding any noise to the variables or further distributional assumptions and allow valid statistical inferences on the parameters of interest. The magnitudes and signs of the estimated parameter vector can be interpreted in the same way as the classical regression coefficients. However, the main drawback is that the estimated parameters are not point identified but rather up to a set to which the parameter vector belongs. In many empirical applications, this set may be too wide to be economically meaningful. Our method circumvents the problem of using interval censored variables and shows that if a multiple discretization scheme is applicable, this additional information can be used to get point identification along with more precise estimates.
%
%
\par
The paper is organized as follows: Section \ref{sec:disc_problem} introduces the discretization problem, data privacy, the parameters of interest, and also, the notation. Section \ref{sec:splitsampling} describes the proposed split sampling approach. We discuss the \textit{shifting} method in detail and show convergence in distribution to the unknown underlying distribution. In Section \ref{sec:estimation}, we identify parameters in a multivariate linear regressions and derive the asymptotic properties of the least square estimator for the cases of discretized variables on the right or left hand side. We also outline the case of discretization on both side of a regression model. Section \ref{sec:nonlinear} discusses how to extend our method to nonlinear models and infer implications with panel data and fixed-effect type of estimators. Section \ref{sec:evidences} presents some Monte Carlo evidence along with an empirical application on the Australian gender wage gap. Section \ref{sec:conclusion} concludes.

%
\section{The Discretization Problem}
\label{sec:disc_problem}
Consider $Z \sim f_Z(z;a_l, a_u)$ an i.i.d. random variable, where $f_Z(z;a_l,a_u)$ denotes the probability density function (pdf) with support $[a_l,a_u]$, where $a_l, a_u \in \mathbb{R}, \, a_l < a_u$. We assume that $f_Z(\cdot)$ is unknown and continuous. Let $Z_i$ be $i=1,\ldots, N$, realizations of $Z$. Instead of providing the sensitive variable $Z_i$, one observes a discretized version $Z^*_i$, through the following discretization process:
\begin{equation} 
\label{eq:discretize_z}
    \mathcal{M}^0(Z_i) = Z^*_{i} = 
    \begin{cases}
		v_1 & \quad \text{if } c_0 \leq Z_{i} < c_1 \quad \text{or} \quad  Z_{i} \in \mathcal{C}_1 = \left[c_0,c_1\right)  \,\, \text{1st interval} \\
		\vdots  & \quad \vdots \\
            v_m  & \quad \text{if } c_{m-1} \leq Z_{i} < c_m \quad \text{or} \quad  Z_{i} \in \mathcal{C}_m = \left[c_{m-1},c_m\right)  \\
        \vdots  & \quad \vdots \\
		v_M  & \quad \text{if } c_{M-1} \leq Z_{i} < c_{M} \quad \text{or} \quad  Z_{i} \in \mathcal{C}_{M} = \left[c_{M-1},c_M\right) \,, \\
        & \qquad \qquad \qquad \qquad \hfill\text{last interval} \\
    \end{cases}
\end{equation}
where $v_m \in \mathcal{C}_m$, $m=1 , \dots , M$ is the assigned value for each interval, and $\mathcal{M}^0$ is the discretization mechanism. The value $v_m$ can be a measure of centrality (e.g., midpoint) or an arbitrarily assigned value within its interval. $M$ denotes the number of intervals that are \textit{known}. Also, the discretization intervals $\mathcal{C}_m$, are set by the data provider or the survey maker, and we take them as given. It is helpful to see the discretization intervals as independent from the random variable $Z$, but this is not a necessary requirement in our case. For simplicity, we use the terms interval and class interchangeably for $\mathcal{C}_m$.
\subsection{Data privacy}
In differential privacy, there is a well-established notion for data protection called $\varepsilon$-differential privacy (\citealt{dwork2006differentialprivacy}, \citealt{dwork2006calibrating}). The idea is the following: one changes a single observation in a dataset and wants to see how it impacts the publicly provided information. $\varepsilon$-differential privacy captures this notion as a ratio of probabilities. Let $\mathcal{Z}$ be an arbitrary sample taken from $f_Z(z;a_l,a_u)$, and $\mathcal{M}(\cdot)$ the privatization mechanism, which privatizes/discretizes $\mathcal{Z}$ into $\mathcal{Z}^*$ for public release. Consider two neighboring realizations of $\mathcal{Z}$; $\mathrm{z}$, and $\mathrm{z}'$ that are different in just one observation. In our case it means to drop one observation from $\mathrm{z}$ to get $\mathrm{z}'$.\footnote{In differential privacy it may also mean to add different noise to one observation.} Consequently, $\varepsilon$-differential privacy requires that the ratio of the probability of any privatized value given one sample to the probability of it given the other sample is upper bound by $e^\varepsilon$, that is:
\begin{equation}
\label{eq:eps_privacy}
    \sup_{\mathrm{z},\mathrm{z}'} \, \sup_{\mathcal{C}_m} \frac{\Pr\left[\mathcal{M}(\mathcal{Z}) \in \mathcal{C}_m | \mathcal{Z} = \mathrm{z} \right]}{\Pr\left[\mathcal{M}(\mathcal{Z}) \in \mathcal{C}_m | \mathcal{Z} = \mathrm{z}' \right]} \leq e^\epsilon ,
\end{equation}
\noindent where $\varepsilon \geq 0$ is a privacy factor. Note that the $\Pr[\cdot]$ refers to the randomness implied in the assignment mechanism $\mathcal{M}(\cdot)$. If $\varepsilon$ is small, we say privacy protection is high. Typically, the literature set $\varepsilon$ to 1 or 2. When the numerator and denominator are zero, the convention is to define them as 0. In the case when the denominator is  zero, but the numerator is not, we face information leakage. To address this case, the literature uses the ``approximate differential privacy'', given by:
\begin{equation}
\label{eq:approx_privacy}
    \sup_{\mathrm{z}} \, \sup_{\mathcal{C}_m}\Pr\left[\mathcal{M}(\mathcal{Z}) \in \mathcal{C}_m | \mathcal{Z} = \mathrm{z} \right] \leq e^\epsilon \sup_{\mathrm{z}'} \, \sup_{\mathcal{C}_m} \left( \Pr\left[\mathcal{M}(\mathcal{Z}) \in \mathcal{C}_m | \mathcal{Z} = \mathrm{z}' \right] \right) + \delta ,
\end{equation}
\noindent where $\delta$ gives the probability of information leakage, typically a small value. 
\par
\textit{Remark 1}: The use of known fixed discretization intervals $\mathcal{C}_m$ allows to hide identity.
\par
\textit{Remark 2}: In our setup, the discretization intervals' size $(M)$ mitigates the trade-off between privacy and accuracy. To illustrate this, the size of the discretization interval is $||\mathcal{C}_m|| = c_{m} - c_{m-1}$. One end of the spectrum is an infinite number of interval; $M\to \infty \implies ||\mathcal{C}_m|| \to 0$, so we observe each sensitive value in its original form without providing any privacy. On the other hand, if we use only one interval, $M=1 \implies ||\mathcal{C}_m|| = a_u - a_l$, which means $Z^*_i$ takes only one value. This case implies the numerator and the denominator are always the same, thus $\varepsilon = 0$.
\par
\textit{Remark 3}: In our approach $\mathcal{C}_m$ are set and known, thus for a given discretization mechanism $\mathcal{M}(\cdot)$, one can calculate $\varepsilon$ and $\delta$ using Equation (\ref{eq:approx_privacy}).
\subsection{The parameter of interest}
\par
We are interested in the regression parameter $\beta$, where $g(\cdot)$ is a known continuous function,
\begin{equation}\label{eq:identification-scalar-DGP}
    \e\left[Y|X\right] = g(X;\beta) \,.
\end{equation}
\noindent
For the sake of simplicity, we introduce our notation in scalar terms; $\beta$ stands for a scalar parameter belonging to a subset of a compact finite-dimensional space $(\mathcal{B} \subset \R)$, while $Y, X$ are scalars. Facing the discretization of a variable can affect the identification of $\beta$ in \textit{three} ways. In the first and simplest case, the explanatory variable $(X^*)$ is discretized; the second possibility is when the outcome variable $(Y^*)$ is discretized; and lastly, where both variables are discretized $(Y^*, X^*)$.
\par
To show an example for the identification problem, let us consider the following model:
\begin{equation}
    Y = X \beta + u \,.
\end{equation}
\noindent Now, instead of observing $X$, we observe the discretized version of $X^*$, resulting in,
\begin{equation}
    Y = X^* \beta + \left(X-X^*\right) \beta + u \,.
\end{equation}
\noindent Under the usual assumptions, the \textit{naive} OLS estimator for $\beta$ results in $\hat{\beta}^n
    = \left ( X^{*\prime} X^* \right )^{-1} X^{*\prime} Y
    = \left( X^{*\prime} X^* \right )^{-1} X^{*\prime} X \beta + \left( X^{*\prime} X^* \right )^{-1} X^{*\prime} u\,$,
that is in general $\hat{\beta}^n \not\to \beta$ as $X^{*} \neq X$. Furthermore, the effect of the discretization is contagious; thus, if one uses further variable(s) $W$ with parameter(s) $\gamma$, that are correlated with $X$, then $\hat\gamma^n \not\to \gamma$. This result implies that the practice of assigning arbitrary values (e.g., midpoints) to the discretized variable leads to biases in the parameter estimates even when the discretized variables are controls. Note that the size of the bias depends on the underlying distribution of $X$ and the discretization process.\footnote{We provide more detailed discussion on this topic with proofs in \href{https://regulyagoston.github.io/papers/MSV\_online\_supplement.pdf}{online supplement}, Section 5.}
\par
In fact, this phenomenon is known in the literature, and \cite{manski2002inference} shows that with discretized variables, only set identification is possible. However, in this paper, instead of applying set identification, we show that if the marginal distribution of the discretized variable is \textit{known}, then under some mild conditions, we can point identify $\beta$ while preserving the discretization process.
%

\section{Split Sampling} 
\label{sec:splitsampling}

The key to our approach is the use of split sampling, which recovers the marginal distribution of $Z$ even if it is discretized. Split sampling consists of two steps: 1) discretize the sensitive variable into multiple samples; 2) combine the samples by using the information on the known interval boundaries. Although the idea is simple, it requires some heavy notation for rigorous proofs. Thus, first, let us take a very simple illustrative example, and then we discuss the univariate case in general. As a last step, we extend our approach to the multivariate models.
\par
Let $Z$ be a random variable with the following pdf, with support boundaries $a_l = 0$ and $a_u=4$,
\begin{equation}
    \Pr(Z \in [a,b)) =
    \begin{cases}
        & 0.5, \,\, \text{if } a=0,\, b=1 \\    
        & 0.3, \,\, \text{if } a=1,\, b=2\\    
        & 0.2, \,\, \text{if } a=2,\, b=4.\\      
    \end{cases} 
\end{equation}
\par
\noindent
As the first step, let us discretize $Z$ with two intervals $M=2$ and let us use two discretization schemes or split samples. The first split sample discretizes the $Z$ as $[0, 2)$ and $[2, 4]$, while in the second split sample $[0, 1)$ and $[1, 4]$. The top section of Figure \ref{fig:splitsample} referred to as ``split samples'' visualizes this setup. By taking these two separate sets of information, we cannot recover the true probabilities for $Z$.
\par
\begin{figure}[H]
    \label{sec:splitsample}
	\centering
	\begin{tikzpicture}
	\draw[thick] (0,0) -- (8,0);
	\foreach \x in {0,4,8}
	\draw[thick] (\x cm , 0.1) -- ( \x cm , -0.1 );
	\foreach \x in {0,4,8}
	\FPeval{\result}{\x/2}
	\FPround{\result}{\result}{0}
	\draw (\x,0) node[below=3pt] {\result};
	\draw (2,0) node[above=3pt]{\scriptsize $\Pr\left[Z^{(1)} \in [0,2)\right] = 0.8$};
	\draw (6,0) node[above=3pt]{\scriptsize $\Pr\left[Z^{(1)} \in [2,4]\right] = 0.2$};
	\draw[thick] (0,-2) -- (8,-2);
	\foreach \x in {0,2,8}
	\draw[thick] (\x cm , -2.1) -- ( \x cm , -1.9 );
	\foreach \x in {0,2,8}
	\FPeval{\result}{\x/2}
	\FPround{\result}{\result}{0}
	\draw (\x,-2) node[below=3pt] {\result};
	\draw (1,-2) node[above=3pt]{\scriptsize $\Pr\left[Z^{(2)} \in [0,1)\right] = 0.5$};
	\draw (5,-2) node[above=3pt]{\scriptsize $\Pr\left[Z^{(2)} \in [1,4]\right] = 0.5$};
        \draw [decorate,decoration={brace,amplitude=5pt}]
                (8.5,1) -- (8.5,-3) node[midway,xshift=2.5em,rotate=-90]{Split samples};
	\draw[thick] (0,-6) -- (8,-6);
	\foreach \x in {0,2,4,8}
	\draw[thick] (\x cm , -6.1) -- ( \x cm , -5.9 );
	\foreach \x in {0,2,4,8}
	\FPeval{\result}{\x/2}
	\FPround{\result}{\result}{0}
	\draw (\x,-6) node[below=3pt] {\result};
	\draw (2,-6) node[above=32pt,rotate=25]{\scriptsize $\Pr\left[Z^{WS} \in [0,1)\right] = 0.5$};
	\draw (4,-6) node[above=32pt,rotate=25]{\scriptsize $\Pr\left[Z^{WS} \in [0,1)\right] = 0.3$};
	\draw (7,-6) node[above=32pt,rotate=25]{\scriptsize $\Pr\left[Z^{WS} \in [2,4]\right] = 0.2$};
        \draw (9.5,-5.75) node {Working};
        \draw (9.7,-6.25) node {sample};
	\draw[>=latex ,->, ultra thick, dashed] (4,-2.5) to (4,-4);
    \end{tikzpicture}
    \caption{Illustrative example for split sampling}
    \label{fig:splitsample}
\end{figure}
\par
\noindent
However, in the second step, we combine information from the two split samples into the ``working sample'', which is represented in the bottom part of Figure \ref{fig:splitsample}. In this simple example, the working sample has $3$ intervals, with boundaries $[0, 1)$, $[1, 2)$, and $[2, 4]$. By combining the information from the split samples to the working sample, we can deduce the probabilities of the non-observed intervals.
\par
Now, let us generalize our method.
First, create multiple samples ($s=1,\dots,S$) and discretize them through different schemes while fixing the number of intervals ($M$) within each split sample.
Discretization for the $s$-th split sample looks exactly as the problem introduced in Equation (\ref{eq:discretize_z}), with the only difference being that the interval boundaries are different,\footnote{To simplify the notation, we use instead of $Z^{*(s)}$ simply $Z^{(s)}$.}
\begin{equation} \label{eq:split samples}
    \mathcal{M}(Z_i,s) = Z^{(s)}_{i} = \begin{cases}
		v_1^{(s)}  & \quad \text{if } Z_{i} \in \mathcal{C}_1^{(s)} = [c_0^{(s)},c_1^{(s)}) , \\
		\vdots  & \quad \vdots \qquad \qquad \qquad \text{1st interval for split sample } s, \\
            v_m^{(s)}  & \quad \text{if }  Z_{i} \in \mathcal{C}_m^{(s)} = [c_{m-1}^{(s)},c_m^{(s)}), \\
        \vdots  & \quad \vdots \\
		  v_M^{(s)}  & \quad \text{if } Z_{i} \in \mathcal{C}_{M}^{(s)} = [c_{M-1}^{(s)},c_M^{(s)}], \\
          & \qquad \qquad \qquad \qquad \text{last interval for split sample } s. \\
    \end{cases}
\end{equation}
The number of observations across split samples $(N^{(s)})$ can be the same or, more likely, different.
\par
As a second step, we combine information into the ``\textit{working sample}''.
The construction of the working sample's interval is simply given by the union of unique interval boundaries within each split sample,
\begin{equation} \label{eq:ws_interval_boundaries}
    \bigcup_{b=0}^B \mathcal{C}_b^{WS} = \bigcup_{s=1}^S\bigcup_{m=0}^M \mathcal{C}_m^{(s)}\,.
\end{equation}

\noindent
One can generalize the approach with any type of distribution. Let us write the probability for a random observation to be in a given class of a split sample,
\begin{equation}
    \Pr\left( Z \in \mathcal{C}^{(s)}_m \right) = \Pr( Z \in \mathcal{S}_s ) \int_{c^{(s)}_{m-1}}^{c^{(s)}_m} f_Z(z)\mathrm dz ,
\end{equation}
where $\mathcal{S}_s$ denotes the $s$-th split sample. We can express, the probability of an observation falling into the working sample's $\mathcal{C}^{WS}_b$ interval as
\begin{equation}
    \label{eq:probWS}
    \Pr \left( Z \in \mathcal{C}^{WS}_b \right) =
        \sum_{s=1}^S \Pr( Z \in \mathcal{S}_s ) \sum_{m=1}^M \Pr \left( Z \in \mathcal{C}^{WS}_b \mid Z \in \mathcal{C}^{(s)}_m \right) \int_{c^{(s)}_{m-1}}^{c^{(s)}_m} f_Z(z)\mathrm dz \,.
\end{equation}
\noindent In practice, we recommend splitting the sample size into equal parts, thus $\Pr( Z \in \mathcal{S}_s ) = 1/S$ and use uniform probability for $\Pr \left( Z \in \mathcal{C}^{WS}_b \mid Z \in \mathcal{C}^{(s)}_m \right)$. Then Equation (\ref{eq:probWS}) simplifies to
\begin{equation}
    \Pr \left( Z \in \mathcal{C}^{WS}_b \right) =
        \frac{1}{S}\sum_{s=1}^S \sum_{\substack{ m \\ \text{ if } \mathcal{C}^{WS}_b \subset \mathcal{C}^{(s)}_m} } \frac{c^{WS}_b-c^{WS}_{b-1}}{c^{(s)}_m-c^{(s)}_{m-1}} \int_{c^{(s)}_{m-1}}^{c^{(s)}_m} f_Z(z)\mathrm dz \,.
\end{equation}
%

\subsection{The shifting method}
\label{sec:shifting}
The shifting method is \textit{an} application of split sampling. It uses an equal distance discretization scheme as the benchmark and shifts the boundaries of the intervals with a certain value. The interval widths for the different split samples remain the same, except for the boundary intervals. As one increases the number of split samples, the size of the shift becomes smaller and smaller. This enables the mapping of the marginal probability distribution at the limit. First, we introduce the discretization process proposed for the shifting method, and then, as a second step, we show how to use the realizations of this discretization to create a synthetic variable that maps the original variable's marginal distribution.
\par
\indent
Figure \ref{fig:shifting_rhs} shows an illustrative example for split samples with $S=4$ and $M=4$ classes.
\par

\begin{figure}[H]
    \centering
        \begin{tikzpicture}
        \foreach \y in {0,1,2,3}
        \draw[thick] (0,-\y) -- (6,-\y);
        \foreach \y in {0,1,2,3}
        \foreach \x in {0,6}
        \draw[thick] (\x,-\y+0.1) -- (\x,-\y-0.1);
        \foreach \y in {0,1,2,3}
        \foreach \x in {0,6}
        \draw (\x,-\y) node[below=3pt] {\x};
        \foreach \y in {0,1,2,3}
        \foreach \x in {\y/2,2+\y/2,4+\y/2}
        \draw[thick] (\x,-\y+0.1) -- (\x,-\y-0.1);
        \foreach \x in {2,4}
        \draw (\x,0) node[below=3pt] {\x};
        \foreach \y in {1,2,3}
        \foreach \x in {\y/2,2+\y/2,4+\y/2}
        \FPeval{\result}{\x}
        \FPround{\result}{\result}{1}
        \draw (\x,-\y) node[below=3pt] {\result};
        \draw [decorate,decoration={brace,amplitude=10pt}, thick] (6.5,0) -- (6.5,-3);
	   \draw[thick] (8.2,-1.15) node {split samples};
	   \draw[thick] (8.2,-1.75) node {$S=4, M=4$};
    	\draw[thick] (0,-5) -- (6,-5);
    	\foreach \x in {0,0.5,1,1.5,2,2.5,3,3.5,4,4.5,5,5.5,6}
    	\draw[thick] (\x,-5.1) -- (\x,-4.9);
    	\foreach \x in {0,0.5,1,1.5,2,2.5,3,3.5,4,4.5,5,5.5,6}
    	\draw (\x,-5) node[below=3pt] {\x};
    	\draw[thick] (8.2,-5) node {Working sample};
    	\draw[thick] (8.2,-5.5) node {$B=(M-1) \times S$};
	\end{tikzpicture}
	\caption{Example for the shifting method}
	\label{fig:shifting_rhs}
\end{figure}
\noindent
To derive the properties of the shifting method, we define the uniform class widths for the first split sample as $\frac{a_u-a_l}{M-1}$. We split these intervals into $S$ parts, thus we can shift the class boundaries $S$ times. This defines the size of the shift as
\begin{equation}
    h = \frac{a_u-a_l}{S(M-1)}\,.
\end{equation}
This implies that the number of intervals in the working sample is $B = S(M-1)$.
The boundary points for each split sample are given by, \footnote{In the \href{https://regulyagoston.github.io/papers/MSV\_online\_supplement.pdf}{online supplement} under Section 1, Algorithm~A1 summarizes how to construct split samples using the shifting method.}
\begin{equation}
    c^{(s)}_m = 
    \begin{cases}
        a_l, & \text{ if } m = 0 ,\\
        a_l + ( s - 1 )\frac{a_u-a_l}{S(M-1)} + ( m - 1 ) \frac{a_u-a_l}{M-1} & \text{ if } 0 < m < M ,\\
        a_u, & \text{ if } m = M . \\
    \end{cases}
\end{equation}
\noindent
As the second step, we introduce a synthetic random variable $Z^\dagger$ that maps the underlying marginal distribution for $Z$. Observations from a particular interval in the split sample $s$ can end up in several candidate intervals from the working sample $\mathcal{C}^{WS}_b$. To make it more transparent, we define the discretization intervals $\mathcal{C}^{(s)}_b$ as sets of $\mathcal{C}^{WS}_b$,
\begin{equation}
    \mathcal{C}^{(s)}_m = 
    \begin{cases}
        \{\emptyset\}, & \text{ if, } s = 1 \text{ and } m = 1 ,\\
        \bigcup_{b=1}^{s-1}\{\mathcal{C}^{WS}_b\}, & \text{ if, } s \neq 1 \text{ and } m = 1 , \\
        \bigcup_{b=s-1+(m-2)(M-1)}^{s-1+(m-1)(M-1)}\{\mathcal{C}^{WS}_b\}, & \text{ if, } 1 < m < M , \\
        \bigcup_{b=B-S+s-1}^{B} \{\mathcal{C}^{WS}_b\}, & \text{ if } m = M .
    \end{cases}
\end{equation}
We create the synthetic variable $Z^\dagger$ by assigning each discretized observation from its split sample $\mathcal{C}^{(s)}_m$ to one of its corresponding working sample's intervals $\mathcal{C}^{WS}_b$ with uniform probability. Note that this assignment mechanism does not assume any distributional form for $Z$ itself, but it randomly assigns each value of $Z^{(s)}$ to the working sample's related interval with uniform probability to get $Z^\dagger$.
The random assignment mechanism can be written as
\begin{equation}
        \Pr \left( Z^\dagger \in \mathcal{C}^{WS}_b | Z^{(s)} \in \mathcal{C}^{(s)}_m \right) = 
        \begin{cases}
            1, & \text{if } s=1 \text{ and } m=1 ,\\
            1/(s-1), & \text{if } s \neq 1 \text{ and } m=1 ,\\
            1/S, & \text{if } 1 < m < M , \,\text{ or }\\
            1/(S-s+1), & \text{if } m=M  .\\
        \end{cases}
\end{equation}
The unconditional probability of $Z^\dagger \in \mathcal{C}^{WS}_b$ for the shifting method, following Equation (\ref{eq:probWS}) and assuming ${\Pr( Z \in \mathcal{S}_s )=1/S}$ now is

\begin{equation}\label{eq:shift_conv}
    \Pr \left( Z^\dagger \in \mathcal{C}^{WS}_b \right) =
        \begin{cases}
            0 , & \text{if } s=1 \text{ and } m=1 ,\\
            \frac{1}{S} \sum_{s=2}^S \frac{1}{s-1} \int_{\mathcal{C}^{(s)}_1|\mathcal{C}^{WS}_b \subset \mathcal{C}^{(s)}_1} f_Z(z)\mathrm dz , & 
                \text{if } s \neq 1 \text{ and } m=1 ,\\
            \frac{1}{S^2} \sum_{s=1}^S \int_{\mathcal{C}^{(s)}_m|\mathcal{C}^{WS}_b \subset \mathcal{C}^{(s)}_m} f_Z(z)\mathrm dz , &  
                \text{if } 1 < m < M ,\\
            \frac{1}{S} \sum_{s=1}^S \frac{1}{S-s+1} \int_{\mathcal{C}^{(s)}_M|\mathcal{C}^{WS}_b \subset \mathcal{C}^{(s)}_M} f_Z(z)\mathrm dz , &
                \text{if } m=M \,. \\
        \end{cases}
\end{equation}
Next, we outline the assumptions under which $Z^\dagger \overset{d}{\rightarrow} Z$  as the number of split samples $(S)$ and observations $(N)$ go to infinity.
\par
\begin{assump}
Let $Z$ be a continuous random variable with probability density function $f_Z(z;a_l,a_u)$, where $a_l,a_u, S, N$ and $\mathcal{C}^{(s)}_m$ are as defined above, then 
\begin{enumerate}
    \item[1.a)] $\frac{S}{N} \rightarrow c$ with $c\in (0,1)$ as $N\rightarrow \infty$. \label{ass:NS_rhs}
    \item[1.b)] $\int^b_a f_Z(z) dz > 0$ for any $(a,b) \subset [a_l,a_u]$.  \label{ass:density_rhs}
\end{enumerate}
\end{assump}
\noindent 
Assumption \ref{ass:NS_rhs}.a) ensures that the number of observations is always higher than the number of split samples.\footnote{Note, that we can decrease $\mathrm{c}$ to be arbitrarily close to $0$.} This is required to identify each point of $f_Z(z;a_l,a_u)$ through the mapping of split samples to the working set.
Assumption \ref{ass:density_rhs}.b) imposes a mild assumption on the underlying distribution: the support of the random variable is not disjoint within the working samples' interval, thus $\int_{c^{WS}_{b-1}}^{c^{WS}_b} f_Z(z)\mathrm dz > 0 \,,\,\, \forall c^{WS}_b$.
\begin{proposition}\label{prop:shifting}
    Under Assumptions \ref{ass:NS_rhs}.a), \ref{ass:density_rhs}.b) and $\Pr (Z \in \mathcal{S}_s) = 1/S$, 
    \begin{equation}
        \lim_{N,S\rightarrow \infty} F_{Z^\dagger}(a) = F_Z(a) \,, \qquad \forall a \in (a_l,a_u)
    \end{equation}
\end{proposition}

\begin{proof}
    \label{sec:app_shifting_proof}
    Recall the probability of $Z$ falling into the working sample's interval $\mathcal{C}^{WS}_b$,
    \begin{equation} 
    \label{eq:app_probWS}
        \Pr \left( Z \in \mathcal{C}^{WS}_b \right) =
            \sum_{s=1}^S \Pr( Z \in \mathcal{S}_s) \sum_{m=1}^M \Pr \left( Z \in \mathcal{C}^{WS}_b \mid Z \in \mathcal{C}^{(s)}_m \right) \int_{c^{(s)}_{m-1}}^{c^{(s)}_m} f_Z(z)\mathrm dz \,.
    \end{equation}
    For any $c_b^{WS}$, $\exists l \in [1,S]$, $m\in [1,M]$ such that $c_b^{WS} = c^{(l)}_m$ and note that as $S \to \infty$, $N \to \infty$ in such a way that the number of participants in each $s$ is greater than 0. Now consider $\Pr (Z^\dagger < c^{WS}_b )= \Pr (Z^\dagger < c^{(l)}_m)$, given $\Pr(Z\in \mathcal{S}_s)= 1/S$ and using Equation (\ref{eq:app_probWS}) gives
    \begin{equation}
        \Pr (Z^\dagger < c^{(l)}_m ) = \frac{1}{S}\sum^S_{s=1} \Pr (Z < c^{(l)}_m | Z<c^{(s)}_m) \Pr (Z<c^{(s)}_m). 
    \end{equation}
    \noindent The summation over the different classes in Equation (\ref{eq:app_probWS}) is being replaced by the cumulative probabilities, where we use the fact that no value greater than $c^{(l)}_m$ is represented in the working sample $c^{WS}_b$. Under the shifting method, $c^{(s)}_m \leq c^{(l)}_m$ for $s<l$ and using the definition of conditional probability gives 
    \begin{equation}
    \begin{aligned}
        \Pr (Z^\dagger < c^{(l)}_m ) =& \frac{1}{S}\sum^S_{s=1} \Pr (Z < c^{(l)}_m, Z<c^{(s)}_m)\\
        =& \frac{1}{S} \sum^l_{s=1} \Pr (Z < c^{(l)}_m, Z<c^{(s)}_m) + \frac{1}{S}\sum^S_{s=l+1}  \Pr (Z < c^{(l)}_m, Z<c^{(s)}_m)\\
        =& \frac{1}{S} \sum^l_{s=1} \Pr ( Z < c_m^{(s)} ) + \frac{1}{S} \sum^S_{s=l+1} \Pr (Z < c^{(l)}_m ). 
    \end{aligned}
    \end{equation}
    \noindent The last line follows from the fact that $\Pr(Z<a_1, Z<a_2) = \Pr (Z<a_1) $ if $a_1<a_2$, and the construction of the shifting method allows us to always disentangle the two cases. Since $l$ is fixed,

    \begin{equation}
    \begin{aligned}
       \Pr (Z^\dagger < c^{(l)}_m) =& \frac{S-l-1}{S} \Pr ( Z < c^{(l)}_m )  + \frac{1}{S} \sum^l_{s=1} \Pr (Z<c^{(s)}_m) \\
       \lim_{N,S\rightarrow \infty} \Pr (Z^\dagger < c^{(l)}_m) =& \Pr ( Z< c^{(l)}_m ).
    \end{aligned}
    \end{equation}
    \noindent where $\Pr(Z^\dagger < a) = F_{Z^\dagger}(a)$ and $\Pr (Z<a) = F_Z(a)$. \qedhere
\end{proof}

%

\subsection{Multivariate case}
\label{sec:multivar_defs}

Let us generalize our results to the multivariate case. Let $\bfZ = (\bfz_1,\dots,\bfz_p,$ $\dots,\bfz_P)$ distributed as $\bfZ \sim \bff_{\bfZ}(\bfz;\bfa_l,\bfa_u)$, i.i.d. random variables with $P$ dimensions, where $\bff_{\bfZ}(\cdot)$ denotes the (joint) pdf with $\bfz = (z_1,\dots,z_p,\dots,z_P)$,\footnote{Note here that $\bfz$ is not a row from $\bfZ$, but represents a vector that is used with the distribution $\bff_{\bfZ}(\cdot)$} and known support  $\mathbb{A}^P = [a_{1,l}, a_{1,u}] \times \ldots \times [a_{P,l}, a_{P,u}] \subseteq \R^P$. We assume $\bff_{\bfZ}(\cdot)$ is unknown and continuous. Instead of observing $\bfZ$, we observe $\bfZ^*$ through a (multivariate) discretization process similar to Equation (\ref{eq:discretize_z}). The difference is that instead of intervals, we use a multi-dimensional grid,\footnote{One-by-one discretization does not work in our case as it will not be able to recover the joint distribution of $\bff_{\bfZ}(\cdot)$. See more in \href{https://regulyagoston.github.io/papers/MSV\_online\_supplement.pdf}{online supplement}, Section 1.3.} where each grid is defined as $\bf\mathcal{C}_m = [c_{1,m-1}, c_{1,m}) \times \ldots \times [c_{P,m-1}, c_{P,m})$ with $c_{m,p}$ being the $m$'th grid point for variable $p$. Note that $c_{p,0} = a_{p,l}$ and $c_{p,M} = a_{p,u}$, for $p=1,\ldots, P$, thus the first boundary vector and last boundary vector contains the boundaries of the distribution's support. $\bfv_m \in \bf\mathcal{C}_m$ is the assigned vector for each grid, that can be a measure of centrality (e.g., focus point) or other point within the grid. In the case of $P=2$, we have a 2D grid that defines the discretization scheme, and each $\bf\mathcal{C}_m$ stands for a rectangle with $M^2$ different partitions.
\par
The shifting method in the multivariate case utilizes the initial discretization scheme with $\bf\mathcal{C}_m$ and shifts the boundary points of the grids with $h_p = \frac{a_{p,u} - a_{p,l}}{S(M-1)}$ for each $p=1,\ldots, P$. This results in $(S(M-1))^P$ different grids in the working sample, which is much larger than the initial $M^P$. Changing the grids' boundary points allows for mapping the underlying distribution using the same inductive logic to prove convergence in distribution. More formally, first consider the following assumptions:

\begin{assump}
Let $\bfZ$ be a $P\times 1$ vector of continuous random variables with joint probability density function $f_\bfZ(\bfz;\mathbb{A}^P)$, where $\mathbb{A}^P, S, N$ and $\bf\mathcal{C}^{(s)}_m$ are as defined above, then 
\begin{enumerate}
    \item[2.a)] $\frac{S}{N} \rightarrow c$ with $c\in (0,1)$ as $N\rightarrow \infty$. \label{ass:NS_rhs_multi}
    \item[2.b)] $\int_{\mathbb{B}^P} f_\bfZ(\bfz) d\bfz > 0$ for any $\mathbb{B}^P \subseteq \mathbb{A}^P$.  \label{ass:density_rhs_multi}
\end{enumerate}
\end{assump}
\noindent 
\begin{proposition}\label{prop:shifting_multi}
    Under Assumptions \ref{ass:NS_rhs_multi}.a), \ref{ass:density_rhs_multi}.b) and $\Pr (\bfZ \in \mathcal{S}_s) = 1/S$, 
    \begin{equation}
        \lim_{N,S\rightarrow \infty} F_{\bfZ^\dagger}(\bfa) = F_\bfZ(\bfa) \,, \qquad \forall \bfa \in \mathbb{A}^P
    \end{equation}
\end{proposition}
\noindent See the complete proof in the \href{https://regulyagoston.github.io/papers/MSV\_online\_supplement.pdf}{online supplement} under Section 1.1.

%
\subsection{Some further remarks}

\subsubsection*{Speed of convergence}

We investigated the speed of convergence for $F_{Z^\dagger}(\cdot) \to F_Z(\cdot)$ as a function of the number of split samples $(S)$ to get insight into the choice of $S$ in finite samples for the univariate case. In general, the rate of convergence is $1/S$, except around the boundaries of the support. Boundaries of the support are defined as: $c^{(1)}_1 + (c^{(S)}_2-c^{(1)}_1)$ for the lower bound, and $c^{(1)}_M + (c^{(1)}_{M-1}-c^{(1)}_M)$ for the upper bound. In these regions, the speed of convergence is $\frac{\log S}{S}$. This result means that even with small values of $S$, one can get a decent mapping of the underlying marginal distribution of $Z$.\footnote{We show the detailed derivations for the speed of convergence in \href{https://regulyagoston.github.io/papers/MSV\_online\_supplement.pdf}{online supplement}, under Section 1.2.}

\subsubsection*{Unbounded support}

Let us consider the case when $Z$ has unbounded support: $a_l = -\infty$ or $a_u = \infty$ or both. If we set the first and last intervals' boundary points to infinity, we are facing censoring. As will be discussed in Section \ref{sec:nonlinear}, one can extend our approach for censoring by using the continuous mapping theorem. However, here we propose a simpler solution: truncate observations that fall into the first or last interval: $\mathcal{C}^{(s)}_1, \mathcal{C}^{(s)}_M, \,\, \forall s$. Although this method loses observations and important units, the flip side is that for the truncated distribution now all previous results apply.

\subsubsection*{Data privacy}

The $\varepsilon$-differential privacy can also be used for split sampling methods. The shifting method introduces a random assignment of each unit to a specific split sample (with a uniform probability of $1/S$). The user receives both the discretized value and the discretization scheme.\footnote{Creating the synthetic variable $Z^\dagger$ can be done by the user; no additional information is gained or provided from that step.}
\par
Realized interval sizes are specific to the shifting method. In general, $||\mathcal{C}^{(s)}_m|| = \frac{a_u-a_l}{M-1},$ $\forall m = 2,\dots, M-1$, thus it does not depend on the number of split samples, only on the number of intervals $M$. However, at the lower and higher bounds, the interval sizes are different: $||\mathcal{C}^{(s)}_m|| = (s-1) \frac{a_u-a_l}{S(M-1)},\,\, \forall s,\, \text{ with } m=1 \text{ or } m=M$. This implies that around the boundaries, the number of split samples influences the size of the intervals, thus $\varepsilon$ depends on both $M$ and $S$, that causes an overall weaker privacy protection. As we show in our empirical application, this does not have real empirical relevance. Note, however, that truncation of the boundary values could provide a simple fix for this issue.

\subsubsection*{Other split sampling methods}

It is possible to work out other discretization schemes that can map the underlying marginal distribution of $Z$. In fact, in the \href{https://regulyagoston.github.io/papers/MSV\_online\_supplement.pdf}{online supplement}, under Section 4, we introduce the \textit{magnifying method}, which offers an alternative method to learn $f_Z(z;a_l,a_u)$. The key to all possible split sampling methods is to show convergence in distribution. Different discretizations may yield different speeds of convergence and other $\varepsilon$-differential privacy properties.

\section{Linear Models with Discretized Variable(s)}
\label{sec:estimation}

This section discusses how to use the synthetic $Z^\dagger$ variable to identify and estimate $\beta$ in a linear regression model. As linear models and OLS have well-known properties, we directly address modelling with multiple regressors. First, we discuss the identification $\beta$ for each model, depending on where the discretization happens. This step helps us to pin down which theoretical expressions to quantify. In all cases, we highlight the discretized variable with an asterisk in the superscript.
\par
Let $\bfy = (Y_1,\ldots,Y_N)'$ a $(N\times 1)$ vector, $\bfX = (\bfx_1,\dots,\bfx_k,\dots,\bfx_K)$ with dimensions $(N \times K)$ where $\bfx_k = (\bfx_{k,1},\ldots,\bfx_{k,N})'$. The $(K \times 1)$ parameter vector of interest is $\bfbeta = (\beta_1,\ldots,\beta_K)'$. We allow further conditioning variables that are not discretized denote by $\bfW = (\bfw_1,\dots,\bfw_j,\dots,\bfw_J)$ i.e., a $(N \times J)$ matrix with $\bfw_j = (\bfw_{j,1},\ldots,\bfx_{j,N})'$ and parameter vector $\bfgamma = (\gamma_1,\ldots,\gamma_J)'$, as this mimics standard practices better. $\bfbeta$ and $\bfgamma$ belong to a subset of a compact finite-dimensional space $(\mathcal{B},\mathcal{G})$.
Let the unknown data-generating process (DGP) be
\begin{equation} \label{eq:ols_model}
    \bfy = \bfX \bfbeta + \bfW \bfgamma + \bfu\,, \qquad \bfu \sim iid \left ( \bfzero, \sigma^2_{\bfu} \bfI \right ) \,,
\end{equation}
where $\bfX$ and $\bfW$ are linearly separable. $\bfu$ is homoskedastic (for simplicity), and we use the usual OLS assumptions: $\e[\bfu|\bfX,\bfW] = 0$ and $\plim_{N\to\infty}\left[\frac{1}{N}\bfX\bfM_{\bfW}\bfX\right]^{-1}=\bfQ$, positive definite matrix, with the usual residual maker $\bfM_{\bfW} = \bfI - \bfW (\bfW'\bfW)^{-1}\bfW'$. In all three discretization cases, we follow a three-step procedure:
\begin{enumerate}
    \item[Step 1.] Point identify $\bfbeta$, through defining conditional expectations using the known discretization intervals $\mathcal{C}^{(s)}_m$.
    \item[Step 2.] Use the synthetic variable $(\bfZ^\dagger)$ to estimate these conditional expectations and then discuss the OLS estimator and its asymptotic properties.
    \item[Step 3.] Use the estimator from Step 2 and replace variable(s) from the linear model with their estimated conditional means. We derive an OLS estimator for $\bfbeta$ and show the asymptotic properties of the estimator in the modified regression, dependent on where discretization happened.
\end{enumerate}
\subsection{Explanatory variable}
\label{sec:est_rhs}
First, let us discuss the case when we observe $\bfX^*$, instead of $\bfX$. We utilize the result of Proposition \ref{prop:shifting_multi}, hence $\bfX^\dagger \overset{d}{\rightarrow} \bfX$, where $\bfX^\dagger = \left(\bfx^\dagger_1,\dots,\bfx^\dagger_k,\dots,\bfx^\dagger_K\right)$ and $\bfx^\dagger_k = (\bfx_{k,1},\ldots,\bfx_{k,N})'$. 
\subsubsection*{Step 1: Identification with \texorpdfstring{$\bfX^*$}{}}
To identify $\bfbeta$, from the DGP defined by Equation (\ref{eq:ols_model}), we need to use the conditional expectations given $\bfX^* \in\bf\mathcal{C}^{(s)}_m$. Applying the conditional expectation operator on our model yields, 
%

\begin{equation}\label{eq:ols_plug_rhs}
    \e\left[\bfy|\bfX^* \in \bf\mathcal{C}^{(s)}_m,\bfW\right] = \e\left[\bfX|\bfX^* \in \bf\mathcal{C}^{(s)}_m\right] \bfbeta + \e\left[\bfW|\bfX^* \in \bf\mathcal{C}^{(s)}_m\right] \bfgamma + \e\left[\bfu|\bfX^* \in \bf\mathcal{C}^{(s)}_m\right] \,,
\end{equation}
\noindent
Note that the classes, $\bf\mathcal{C}^{(s)}_m$ are mutually exclusive for given $s$ along $m$ and $k$.
\par
Let us define $\kappa(s,\bfm)$ a function\footnote{Note that $\bfm = (m_1,\dots,m_K)$ is a vector, defining the grid points in $\bf\mathcal{C}^{(s)}_m$ for each variable $k$.} that maps the conditional expectations of $\bfX$ for given grid $\bf\mathcal{C}^{(s)}_m$. Proposition \ref{prop:shifting_multi} ensures equality between the conditional expectation of the observable $\bfX^\dagger$, and the conditional expectation of the unknown $\bfX$ given $\bfX^* \in \bf\mathcal{C}^{(s)}_m$, which is needed to identify $\bfbeta$,
\begin{equation}
   \kappa(s,\bfm) := \e\left[\bfX | \bfX^* \in \bf\mathcal{C}^{(s)}_m \right] =  \lim_{N,S\to\infty}\e\left[\bfX^\dagger | \bfX^* \in \bf\mathcal{C}^{(s)}_m \right] \,.
\end{equation}
\subsubsection*{Step 2: OLS estimator for \texorpdfstring{$\e\left[\bfX|\bfX^* \in \bf\mathcal{C}^{(s)}_m\right]$}{}}
Let $\bfkappa = (\kappa(1,{\bf1}),\dots,\kappa(S,{\bf1}),\dots,\kappa(S,\bfm),\dots,\kappa(S,\bfM))'$ be a vectorized version of $\kappa(s,\bfm)$ with $(SM^K \times 1)$ dimension. We propose a joint estimation for $\bfkappa$ using OLS\footnote{We can estimate all conditional expectations jointly via OLS. Otherwise simple conditional means for a given interval $\bf\mathcal{C}^{(s)}_m$ would yield the same result. Joint estimation however allows to establish asymptotic properties more easily.},
\begin{equation}
    \hat{\bfkappa} = \left( \mathbf{1}'_{ \{ \bfX^\dagger \in \bf\mathcal{C}^{(s)}_m \} } \mathbf{1}_{\{ \bfX^\dagger \in \bf\mathcal{C}^{(s)}_m \} }\right)^{-1} \mathbf{1}'_{ \{ \bfX^\dagger \in \bf\mathcal{C}^{(s)}_m \} } \bfX^\dagger \,\,,
\end{equation}
where $\mathbf{1}_{\{\cdot\}}$ stands for the indicator function and results in a matrix with dimensions $(N\times SM^K)$ with the same ordering as $\bfkappa$. Let us note, $\hat{\bfkappa}$ does not require the actual values of $\bfX$, only $\bfX^\dagger$ and $\bfX^\dagger \in \bf\mathcal{C}^{(s)}_m$. Under standard OLS assumptions,
    $\sqrt{N}\left(\hat{\bfkappa} -  \bfkappa \right) \asim{} \mathcal{N}\left(\mathbf{0},\bfOmega_{\bfkappa} \right)$,
%
\noindent where $\bfOmega_{\bfkappa}$ is the variance-covariance matrix. See the derivations for the proposed estimator in the \href{https://regulyagoston.github.io/papers/MSV\_online\_supplement.pdf}{online supplement}, under Section 2.1.

\subsubsection*{Step 3: OLS estimator for \texorpdfstring{$\bfbeta$}{}}
To get a consistent estimator for $\bfbeta$, let us follow Equation (\ref{eq:ols_plug_rhs}), and define $\ddot\bfX$, which replaces $\bfX^*$ with the corresponding conditional expectations via $\hat\bfkappa$. Furthermore, let $\ddot\bfy = \hat{\e}\left[\bfy|\bfX^* \in \bf\mathcal{C}_m\right]$ and $\ddot\bfW = \hat{\e}\left[\bfW|\bfX^* \in \bf\mathcal{C}_m\right]$ the respective conditional averages.
\begin{equation}
\label{eq:disc_rhs_model}
    \ddot\bfy = \ddot\bfX \bfbeta + \ddot\bfW \bfgamma + \bfe \,,
\end{equation}
\noindent  while $\bfe = (e_1,\dots,e_N)$. Now the OLS estimator for $\bfbeta$ is
\begin{equation}
    \hat{\bfbeta} = \left ( \ddot\bfX^\prime \bfM_{\ddot\bfW} \ddot\bfX \right )^{-1} \ddot\bfX^{\prime}\bfM_{\ddot\bfW} \ddot\bfy\,,
\end{equation}
\noindent where $\bfM_{\ddot\bfW}$ is the usual residual maker using $\ddot\bfW$. Note that $\E[e_i]=0$ for all $i$ since $\ddot{\bfX}$ is a consistent estimate. Moreover, $\E[e_i e_j] = 0$ for $i\neq j$ due to $\bf\mathcal{C}^{(s)}_m$ being mutually exclusive in $m$ and $k$. As long as the discretization is mean independent from the error term: $\e\left[\bfu|\bfX,\bfW,\bfX^*\in\bf\mathcal{C}_m\right] = \e\left[\bfu|\bfX,\bfW\right] = 0$, the OLS is a consistent estimator.
\par
\begin{proof}[Proof of consistency]
\begin{align*}
    \plim_{N,S \to \infty} \hat\bfbeta 
    &= \plim_{N,S \to \infty} \left ( \ddot\bfX^\prime \bfM_{\ddot\bfW} \ddot\bfX \right )^{-1} \ddot\bfX^{\prime}\bfM_{\ddot\bfW} \ddot\bfy \\
    &= \plim_{N,S \to \infty} \left[\left ( \ddot\bfX^\prime \bfM_{\ddot\bfW} \ddot\bfX \right )^{-1} \ddot\bfX^{\prime}\bfM_{\ddot\bfW} \left(\ddot\bfX \bfbeta + \ddot\bfW \bfgamma + \bfe\right)\right]  \\
    &= \bfbeta + \bf0 + \plim_{N,S \to \infty} \left[ \left ( \ddot\bfX^\prime \bfM_{\ddot\bfW} \ddot\bfX \right )^{-1} \ddot\bfX^{\prime}\bfM_{\ddot\bfW} \bfe \right]\\
    &= \bfbeta + \plim_{N,S \to \infty} \left[ \left ( \ddot\bfX^\prime \bfM_{\ddot\bfW} \ddot\bfX \right )^{-1} \ddot\bfX^{\prime}\bfM_{\ddot\bfW} \e\left[\bfu|\bfX,\bfW,\bfX^*\in\bf\mathcal{C}_m\right] \right] \\
    &= \bfbeta . \qquad \qedhere
\end{align*}    
\end{proof}
\noindent The asymptotic distribution of $\hat\bfbeta$ can be derived similarly. As $\bfu \sim iid \left ( \bfzero, \sigma^2_{\bfu} \bfI \right )$ and as $\ddot{\bfX}$ and $\ddot{\bfW}$ are consistent estimates of the conditional means, $\e\left[\bfe\bfe'|\ddot{\bfX},\ddot{\bfW}\right] = \sigma_{e}^2\bfI$. Under assumptions $N,S \to \infty$ and $\e\left[\bfu|\bfX,\bfW,\bfX^*\in\bf\mathcal{C}_m\right] = \e\left[\bfu|\bfX,\bfW\right]$, we get $\sqrt{N}\left(\hat\bfbeta-\bfbeta\right) \asim{} \mathcal{N}\left(\mathbf{0}, \sigma_{e}^2\left(\ddot{\bfX}'\bfM_{\ddot\bfW}\ddot{\bfX}\right)^{-1} \right)$.

\subsection{Outcome variable}
\label{sec:est_lhs}

The second case is when the outcome variable is discretized, thus $\bfy$ is not observed, only $\bfy^*$. To simplify our analysis let us neglect $\bfW$ along with the parameter vector $\bfgamma$ as they can be seen as part of $\bfX$ and $\bfbeta$ respectively.
\subsubsection*{Step 1: Identification with \texorpdfstring{$\bfy^*$}{}}
To identify $\bfbeta$ with discretized $\bfy^*$, let us partition the domain of $\bfX$'s into mutually exclusive partitions, denoted by $\bf\mathcal{D}_l$ and defined by the researcher, such that it is (mean) independent from the error term $\bfu$. We use the same logic as derived in Equation (\ref{eq:ols_plug_rhs}), but instead of conditioning on $\bfX^* \in \mathcal{C}_m$, we use $\bfX \in \bf\mathcal{D}_l$. Similar derivation can be done,
\begin{equation}
\begin{aligned}
    \e\left[\bfy^* | \bfX \in \bf\mathcal{D}_l \right] &= \e \left[ \bfX | \bfX \in \bf\mathcal{D}_l \right] \bfbeta  \\
    \e\left[\e[\bfy | \bfy^* \in \mathcal{C}_m, \bfX \in \bf\mathcal{D}_l] | \bfX \in D_l \right] &= \e \left[ \bfX | \bfX \in \bf\mathcal{D}_l \right] \bfbeta \\
    \sum_l\sum_m\e[\bfy | \bfy^* \in \mathcal{C}_m, \bfX \in \bf\mathcal{D}_l] \Pr[\bfy^* \in \mathcal{C}_m | \bfX \in \bf\mathcal{D}_l] &= \e \left[ \bfX | \bfX \in \bf\mathcal{D}_l \right] \bfbeta\,. \label{eq:id_discy_beta}
\end{aligned}
\end{equation}
$\Pr[\bfy^* \in \mathcal{C}_m | \bfX \in \bf\mathcal{D}_l], \e \left[ \bfX | \bfX \in \bf\mathcal{D}_l \right]$ are known quantities and we aim to learn $\e[\bfy | \bfy^* \in \mathcal{C}_m, \bfX \in \bf\mathcal{D}_l]$. Note that $\e \left[ \bfu | \bfX \in \bf\mathcal{D}_l \right] = 0$ by the definition of $\bf\mathcal{D}_l$.
\par
Let us define $\pi(s,m,\bfl)$ a function that maps the conditional expectations of $\bfy$ for a given interval $\bfy \in \mathcal{C}^{(s)}_m$ and partition $\bfX \in \bf\mathcal{D}_l$. 
\begin{equation}
    \pi(s,m,\bfl) := \e\left[ \bfy | \bfy^* \in \mathcal{C}^{(s)}_m, \bfX \in \bf\mathcal{D}_l \right] = \lim_{N,S\to\infty}\e\left[ \bfy^\dagger | \bfy^\dagger \in \mathcal{C}^{(s)}_m, \bfX \in \bf\mathcal{D}_l \right] \,,    
\end{equation}
where $\bfy^\dagger = (Y^\dagger_1,\ldots,Y^\dagger_N)$ and we used the results of the shifting method to show convergence. 
\subsubsection*{Step 2: OLS estimator for \texorpdfstring{$\e\left[ \bfy | \bfy^* \in \mathcal{C}^{(s)}_m, \bfX \in \bf\mathcal{D}_l \right]$}{}}
Let $\bfpi = \left(\pi(1,1,{\bf1}),\dots,\pi(S,1,{\bf1}),\dots,\pi(S,M,{\bf1}),\dots,\pi(S,M,\bfl),\dots,\pi(S,M,\bfL)\right)'$ be the vectorized version of $\pi(s,m,\bfl)$. Estimating $\bfpi$, via OLS yields,
\begin{equation}\label{eq:estimator_cond_mean_lhs}
    \hat{\bfpi} = \left( \mathbf{1}'_{ \{ \bfy^\dagger \in \mathcal{C}^{(s)}_m, \bfX \in \bf\mathcal{D}_l\} } \mathbf{1}_{\{ \bfy^\dagger \in \mathcal{C}^{(s)}_m, \bfX \in \bf\mathcal{D}_l \} }\right)^{-1} \mathbf{1}'_{ \{ \bfy^\dagger \in \mathcal{C}^{(s)}_m, \bfX \in \bf\mathcal{D}_l \} } \bfy^\dagger \,.
\end{equation}
\noindent
The vector $\hat{\bfpi}$ has dimensions of $(SML^K \times 1)$, $ \mathbf{1}_{\{\cdot\}}$ is the corresponding indicator matrix, and estimation of $\bfpi$ does not require the actual values of $\bfy$, only $\bfy^\dagger$, and observing $\bfy^\dagger \in \mathcal{C}^{(s)}_m$ and $\bfX \in \bf\mathcal{D}_l$. Under the standard OLS assumption,
%
    $\sqrt{N}\left(\hat{\bfpi} - \bfpi \right) \asim{} \mathcal{N}\left(\mathbf{0},\bfOmega_{\bfpi} \right)$,
%
\noindent where $\bfOmega_{\bfpi}$ is the corresponding variance-covariance matrix. See the derivations for the proposed estimator in the \href{https://regulyagoston.github.io/papers/MSV\_online\_supplement.pdf}{online supplement}, under Section 2.2.

\subsubsection*{Step 3: OLS estimator for \texorpdfstring{$\bfbeta$}{}}
As we have shown in our identification strategy, the standard regression model can be rewritten in the form of Equation (\ref{eq:id_discy_beta}). 
Let $\tilde{\bfy}$ be a $(N\times 1)$ vector, where we replace the discretized elements of $\bfy^*$ with the consistent estimates of $\sum_l\sum_m\e[\bfy | \bfy^* \in \mathcal{C}_m, \bfX \in \bf\mathcal{D}_l] \Pr[\bfy^* \in \mathcal{C}_m | \bfX \in \bf\mathcal{D}_l]$, using $\hat\bfpi$ and sample analogues for $\Pr[\bfy^* \in \mathcal{C}_m | \bfX \in \bf\mathcal{D}_l]$.
Let us redefine $\tilde{\bfX} = \hat\e \left[ \bfX | \bfX\in \bf\mathcal{D}_l \right]$, as the conditional sample means for $\bfX$ in partition $\bf\mathcal{D}_l$. With the replaced measures, we get
\begin{equation}
    \tilde{\bfy} = \tilde{\bfX} \bfbeta + \bfnu ,
\end{equation}
\noindent 
where $\bfnu = \left ( \nu_1, \ldots, \nu_N \right )'$. Note that $\E[\nu_i]=0$ for all $i$ since $\tilde{\bfy}$ and $\tilde{\bfX}$ are consistent estimates. Moreover, $\E[\nu_i \nu_j] = 0$ for $i\neq j$ due to that $\bf\mathcal{D}_l$ are mutually exclusive $\forall l,$ and $\E [\nu_i | \tilde{\bfX} ] = 0$ since the partitioning is (mean) independent from the error term and does not affect the sampling error. Let $\hat{\bfbeta} = \left ( \tilde{\bfX}' \tilde{\bfX} \right )^{-1} \tilde{\bfX}' \tilde{\bfy}$, then under similar argument, $\hat{\bfbeta} - \bfbeta = o_p(1)$.
\par 
The asymptotic distribution of $\hat\bfbeta$ can be derived in the same spirit. As $\bfu \sim iid \left( \bfzero, \sigma^2_{\bfu} \bfI \right )$ and $\tilde{\bfX}$ is a consistent estimate of the conditional means, $\e\left[\bfnu\bfnu'|\tilde{\bfX}\right] = \sigma_{\nu}^2\bfI$. The proposed estimator is asymptotically distributed as $\sqrt{N}\left(\hat\bfbeta-\bfbeta\right) \asim{} \mathcal{N}\left(\mathbf{0}, \sigma_{\nu}^2\left(\tilde{\bfX}'\tilde{\bfX}\right)^{-1} \right)$. 

%
%
\subsection{Discretization on both sides}
\label{sec:est_both}

Our last case is when the discretization happens with both outcome $\bfy^*$, and with one or more explanatory variables $\bfX^*$. In this case, we do not need to partition the domain of $\bfX$, but can use the discretization grids $\bf\mathcal{C}^{(s)}_m$ for the explanatory variable. As $\bfy$ is also discretized, we need to partition $\bfW$ with $\bf\mathcal{D}_l$, similarly to the case discussed in Section \ref{sec:est_lhs}. Identification is the same as with discretized outcome variable and the sample estimator for the conditional expectations follows the same logic as the OLS estimator for the $\bfbeta$ parameter. We discuss this case in details in the \href{https://regulyagoston.github.io/papers/MSV\_online\_supplement.pdf}{online supplement}, under Section 3.

\section{Nonlinear Models and Panel Data}
\label{sec:nonlinear}
A possible extension is the use of nonlinear models with one or more variables discretized. Let us consider the following general model:
\begin{equation} \label{eq:nonlinear}
    \bfy = g(\bfX,\bfW; \bfbeta,\bfgamma) + \bfu ,
\end{equation}
\noindent where $g(\cdot)$ denotes a known continuous function. 
\par
The identification procedure is similar to Equations (\ref{eq:ols_plug_rhs}, \ref{eq:id_discy_beta}), with the exception that $g(\e[\cdot;\bfbeta]) \neq \bfbeta g(\e[\cdot])$ in general. As shown with \textit{split sampling} $\bfZ^\dagger \overset{p}{\rightarrow} \bfZ$. If so, with discretized explanatory variable(s) by continuous mapping theorem $g(\bfX^\dagger;\bfbeta)~\overset{p} {\rightarrow}~g(\bfX;\bfbeta)$. Similar argument applies when discretization happens on the left hand side $(\bfy^*)$ or on both sides $(\bfy^*,\bfX^*)$.
\par
As the estimation, let $\hat{\bfbeta}(\bfy, \bfX,\bfW)$ denote a consistent estimator of $\bfbeta$ with $\rho (\bfX,\bfW) =$\\ $\sqrt{N} \left [ \hat{\bfbeta} (\bfy, \bfX, \bfW) - \bfbeta \right ]$ such that $\rho(\bfy, \bfX,\bfW) \overset{d}{\rightarrow} \bff(\bf0, \bfOmega)$. Under the assumptions made earlier, for all cases, the nonlinear quantities should converge too. When the regressors are discretized, we need $\bfX^\dagger \overset{d}{\rightarrow} \bfX \implies \rho (\bfy, \bfX^\dagger,\bfW) \overset{d}{\rightarrow} \rho (\bfy, \bfX,\bfW) $. In case of discretized outcome, $\bfy^\dagger \overset{d}{\rightarrow} \bfy \implies \rho (\bfy^\dagger, \bfX,\bfW) \overset{d}{\rightarrow} \rho (\bfy, \bfX,\bfW)$; while for both variables $(\bfy^\dagger, \bfX^\dagger) \overset{d}{\rightarrow} (\bfy,\bfX) \implies \rho (\bfy^\dagger, \bfX^\dagger,\bfW) \overset{d}{\rightarrow} \rho (\bfy, \bfX,\bfW)$. By the continuous mapping theorem under appropriate regularity conditions these convergences hold. The technical details of these conditions, however, could be an interesting subject of future research.
\par
\bigskip
In case of panel data, the extension of this methodology is relatively straightforward. The most important difference is in the identification. If the interval of an individual does not change over the time periods covered, the individual effects in the panel and the parameter associated with the choice variable cannot be identified separately. The within transformation would wipe out the interval variable as well. When the interval does change over time, but not much, then we are facing weak identification, i.e., in fact very little information is available for identification, so the parameter estimates are going to be highly unreliable.\footnote{See more on this issue in the \href{https://regulyagoston.github.io/papers/MSV\_online\_supplement.pdf}{online supplement}, Section 5.6} 
In the \href{https://regulyagoston.github.io/papers/MSV\_online\_supplement.pdf}{online supplement}, Section 6, we extend our method towards fixed effect type of estimators. This is a good opportunity to raise awareness through a phenomenon that we call the {\it perception effect}. The perception effect is relevant if the discretization happens by surveys. There is much evidence in the behavioral literature that the answers to a question may depend on the way the question is asked (see, e.g., \citealp{Haisley2008}). 
Note, that this is present regardless of whether split sampling has been performed or not. However, with split sampling, there is a way to tackle this issue, much akin to the approach a similar problem has been dealt with in the panel data literature.

%
\section{Simulation and Empirical Evidences}
\label{sec:evidences}

\subsection{Monte Carlo evidence}

For the simulation experiments, we consider a simple univariate linear regression model
\begin{equation}
    Y_i = X'_i\beta+\epsilon_i \, ,
\end{equation}
where we discretize $Y_i$, or $X_i$ or both. We set the parameter of interest $\beta=0.5$ and report the Monte Carlo average bias and its standard deviation in parenthesis 
with $1,000$ repetitions and $N=10,000$ observations.
As the size of the bias depends on mapping the conditional expectations, we use multiple distributions for $\epsilon_i$ or $X_i$. To be more specific, we use ``\textit{Normal}'' (standard normal distribution truncated at $-1$ and $3$), ``\textit{Logistic}'' (standard logistic distribution truncated at $-1$ and $3$), ``\textit{Log-Normal}'' (standard log-normal distribution truncated at $4$ and subtracted $1$ to adjust the support), ``\textit{Uniform}'' (uniform distribution between $-1$ and $3$), ``\textit{Exponential}'' (exponential distribution with rate parameter $0.5$, truncated at $4$ and subtracted $1$) and ``\textit{Weibull}'' (weibull distribution with shape parameter $1.5$ and scale parameter $1$, truncated at $4$ and subtracted $1$). All of these distributions have known finite support between $-1$ and $3$.
We use these distributions to generate $\epsilon_i$ when the outcome or both variables are discretized. In these cases $X_i \sim \mathcal{N}(0,0.25,-1,1)$ (truncated normal distribution between -1 and 1). When discretization happens with the explanatory variable $X_i$, we use the different distributions to generate $X_i$ and use the same truncated normal distribution for $\epsilon_i$. This setup allows us to control for the domain of $Y_i$ in all cases.
For the discretization of the variable(s) in all cases we use $M=5$ and for split sampling $S=10$. We have experimented with different setups; see more in the \href{https://regulyagoston.github.io/papers/MSV\_online\_supplement.pdf}{online supplement}, Section 7.
\par
To estimate $\beta$, our method is shown by the ``\textit{Shifting method}''.\footnote{We used mid-values as observations for the working sample's values $(v^{WS}_b)$.}. As alternatives, we use ``\textit{Mid-point regression}'', a commonly used naive model. This method uses mid-points for $v_m$ with OLS for estimation. When the outcome variable is discretized, we use further popular estimation methods: ``\textit{Set identification}'' which provides estimates for the lower and upper bounds of the parameter set;\footnote{Estimation is based on \cite{beresteanu2008asymptotic}, we use the package by \cite{beresteanu2010stata}.} ``\textit{Ordered probit and logit}'' models\footnote{Note that the estimated maximum likelihood ``naive'' parameters reported here are not designed to recover $\beta$ and to be interpreted in the linear regression sense. Therefore, we call the difference from $\beta$ {\sl distortion} rather than bias.} and ``\textit{Interval regression}''\footnote{We assume a Gaussian conditional distribution to model the censored interval outcome.} as alternative estimators.
We specify our models when applicable to estimate as $Y_i = \alpha + X_i'\beta + \eta_i$, where $\epsilon_i = \alpha + \eta_i$ with $\e (\epsilon_i) = \alpha$, to adjust for those distributions, which does not have zero mean.\footnote{Ordered choice models' implementation in Stata removes the intercept parameter to identify $\beta$.}
\begin{table}[ht!]
    \begin{center}
     \caption{Monte Carlo average bias and standard deviations}
     \label{tab:mc_results}
    \resizebox{.98\textwidth}{!}{%
    \begin{tabular}{l|cccccc}
    	\toprule
    	\toprule
    	\multicolumn{7}{c}{Discretized explanatory variable $(X_i^*)$} \\ \midrule
    	\midrule
    	& Normal & Logistic & Log-Normal & Uniform & Exponential & Weibull \\ 
    	\midrule
    	\multirow{2}{*}{Mid-point regression} 
    	&-0.0252&-0.0101&-0.0174&0.0002&0.0005&-0.0422\\
    	&(0.0057)&(0.0046)&(0.0051)&(0.0040)&(0.0102)&(0.0073)\\ \midrule
    	\multirow{2}{*}{Shifting method $(S=10)$} 
    	&-0.0037&-0.0003&-0.0022&0.0002&0.0023&-0.0015\\
    	&(0.0060)&(0.0046)&(0.0050)&(0.0038)&(0.0094)&(0.0073)\\
    	\midrule
    	\midrule
    	\multicolumn{7}{c}{Discretized outcome variable $(Y_i^*)$} \\
    	\midrule
    	\midrule
    	& Normal & Logistic & Log-Normal & Uniform & Exponential & Weibull \\ 
    	\midrule
    	\multirow{2}{*}{Set identification$^\dagger$} 
    	& $\left[-1.1,1.15\right]$ & $\left[-1.09,1.15\right]$ & $\left[-1.09,1.16\right]$ & $\left[-1.07,1.17\right]$ & $\left[-1.06,1.19\right]$ & $\left[-1.09,1.15\right]$ \\
    	& [(0.02),(0.02)] & [(0.03),(0.03)] & [(0.02),(0.02)] & [(0.03),(0.03)] & [(0.03),(0.03)] & [(0.02),(0.02)] \\\midrule
    	\multirow{2}{*}{Ordered probit$^*$} 
    	& 0.1971 & 0.0688 & 0.2085 & 0.0158 & 0.0986 & 0.4461 \\
    	& (0.0256) & (0.0253) & (0.0262) & (0.0234) & (0.0241) & (0.0295) \\\midrule
    	\multirow{2}{*}{Ordered logit$^*$} 
    	& 0.6509 & 0.3814 & 0.6862 & 0.2379 & 0.4338 & 1.2085 \\
    	& (0.0464) & (0.0455) & (0.0499) & (0.0422) & (0.044) & (0.0546) \\\midrule
    	\multirow{2}{*}{Interval regression} 
    	& 0.0268 & 0.0332 & 0.0371 & 0.0491 & 0.0663 & 0.0397 \\
    	& (0.0198) & (0.0249) & (0.0221) & (0.0271) & (0.0249) & (0.0166) \\\midrule
    	\multirow{2}{*}{Mid-point regression} 
    	& 0.0253 & 0.0322 & 0.0362 & 0.0490 & 0.2077 & 0.0314 \\
    	& (0.0195) & (0.0236) & (0.0216) & (0.0273) & (0.0128) & (0.0157) \\\midrule
    	\multirow{2}{*}{Shifting method} 
    	& -0.0010 & -0.0017 & -0.0010 & -0.0014 & -0.0017 & -0.0003 \\
    	& (0.0211) & (0.0239) & (0.0215) & (0.0271) & (0.0125) & (0.0147) \\
    	\midrule
    	\midrule
    	\multicolumn{7}{c}{Both variables are discretized $(Y_i^*,X_i^*)$} \\
    	\midrule
    	\midrule
    	& Normal & Logistic & Log-Normal & Uniform & Exponential & Weibull \\ 
    	\midrule
    	\multirow{2}{*}{Mid-point regression} 
    	&-0.0853&-0.0788&-0.0752&-0.0635&0.0797&-0.0759\\
    	&(0.0178)&(0.0213)&(0.0190)&(0.0243)&(0.0116)&(0.0137)\\ \midrule
    	\multirow{2}{*}{Shifting $(S=10)$} 
    	&-0.0027&0.0156&0.0104&0.0156&0.0006&0.0108\\
    	&(0.0235)&(0.0269)&(0.0243)&(0.0294)&(0.0132)&(0.0156)\\ 
    	\bottomrule
    \end{tabular}
    }
    \end{center}
    \vspace{-10pt}
    \linespread{1.0}\selectfont
    \scriptsize{
    Row names refers to $X_i$ for ``discretized explanatory variable'' and to $\varepsilon_i$ when the outcome or both variables are discretized. Average bias and the standard deviations of the estimates in parenthesis are reported. $\beta = 0.5$. In case of $Y^*_i$, we have used $L=50$ with equal distances to partition $X_i$. When both variables are discretized, we used $M=M_Y=M_X=5$ and $S=S_Y=S_X=10$. \\
    $^\dagger$ Set identification gives the lower and upper boundaries for the valid parameter set. We report these bounds subtracted from the true parameter; therefore, it should give a (close) interval around zero. \\
    $^*$ Distortion from the true $\beta$ is reported. Ordered probit and logit models' maximum likelihood parameters do not aim to recover the true $\beta$ parameter; therefore we do not to call them biased.}
\end{table}
Table \ref{tab:mc_results} shows our results. 
The shifting method provides smaller average bias than the mid-point regression by magnitudes between 5-50 almost everywhere. The two exceptions can be found in the case of discretized explanatory variable. When we use a uniform distribution, there is no bias for mid-point regression, as the conditional expectation is the mid-point of the interval. The second case is the exponential distribution, where the case is similar, as the curvature of the used distribution is rather flat.\footnote{This result is in line with the theoretical results shown in the \href{https://regulyagoston.github.io/papers/MSV\_online\_supplement.pdf}{online supplement}, Section 5.} The shifting method also outperforms the other competing methods, when discretization happens on the left hand side and provides significant reduction in the bias when both variables are discretized.
\par
We have run several other Monte Carlo experiments and the results are similar: the shifting method outperforms all alternatives. An important result is consistency, as we increase $N$ along with $S$, we get smaller biases. This consistency does not hold for the competing methods. For a detailed discussion see the \href{https://regulyagoston.github.io/papers/MSV\_online\_supplement.pdf}{online supplement}, Section 7.

\subsection{The Australian gender wage gap}
\label{sec:ato_case}
We illustrate our approach with a short study of the Australian gender wage gap from 2017. We obtained wage and socio-economic data from the Australian Tax Office (ATO). The individuals' sample files record a $2\%$ sample of the whole population,\footnote{For details, see ATO's website: \url{https://www.ato.gov.au}.} including the actual wage values. This enabled us to estimate $\beta$ when no discretization happens. In practice, researchers use interval censored data coming from surveys or privatized data. To mitigate this fact, we employ three different equally distanced discretization methods with $M=3,5$, and $10$ and a privatization method that adds Laplacian noise to the data.\footnote{We used the python package \texttt{diffprivlib} by \cite{diffprivlib} for differential privacy implementation.}
\par
Table \ref{tab:gender_coef} shows regression estimates of the gender dummy (1 is female) on log yearly wage while controlling for age, age square, marital status, occupation, and regional variables. The \textit{``Directly observed''} row shows the estimated parameter when the actual values are used for yearly wages. In the following rows we show the estimates when the discretization process is done in equal-sized intervals with $M=3,5,10$. Mid-point regression does not use split samples, while with the shifting method we use $S=10$. In all cases, the yearly wage was discretized, and then we took the log of the discretized mid values as this represents the practice. We report the point estimates along with the standard errors for $\hat\beta$. Furthermore, we report the calculated $\varepsilon$-differential privacy measures based on Equation (\ref{eq:approx_privacy}). Finally, we report an estimate for $\hat\beta$ using differential privacy method, while setting $\varepsilon = 1$. We need to note that the estimates using differential privacy are highly volatile depending on the randomization.
\begin{table}
    \centering
    \caption{Different estimates for gender wage gap based on different discretizations of yearly wages}
    \label{tab:gender_coef}
    \resizebox{\textwidth}{!}{
     \begin{tabular}{l|ccc|ccc|ccc}
    	\toprule
    	\multirow{2}{*}{Directly observed} & \multicolumn{9}{c}{$\hat\beta \qquad \,\,\, SE[\hat\beta]$ } \\ \cline{2-10}\noalign{\vspace{.25em}}
    	& \multicolumn{9}{c}{-0.2261 $\quad$ (0.0023)}\\ 
    	\midrule
    	\midrule
    	\multirow{2}{*}{Discretization methods}  & \multicolumn{3}{c|}{$M=3$} & \multicolumn{3}{c|}{$M=5$} & \multicolumn{3}{c}{$M=10$}  \rule{0pt}{.25em}\\ \cline{2-10}\noalign{\vspace{.25em}}
    	& $\hat\beta$ & $SE[\hat\beta]$ & $\varepsilon$-diff. & $\hat\beta$ & $SE[\hat\beta]$ & $\varepsilon$-diff. & $\hat\beta$ & $SE[\hat\beta]$ & $\varepsilon$-diff. \\ \midrule
    	Mid-point regression
    	& -0.2747 & (0.0033) & 0.0001 
    	& -0.2635 & (0.0031) & 0.0004
    	& -0.2364 & (0.0025) & 0.0011 \\
    	Shifting method ($S=10$)
    	& -0.2214 & (0.0030) & 0.0351 &
    	-0.2348 & (0.0027) & 0.0800 & 
    	-0.2280 & (0.0024) & 0.1178 \\ 
    	\midrule
    	\midrule
    	Differential privacy & \multicolumn{9}{c}{$\hat\beta^{DP} \qquad \,\,\, SE[\hat\beta^{DP}]$ } \\ \cline{2-10}\noalign{\vspace{.25em}}
    	method, with $\varepsilon = 1$ & \multicolumn{9}{c}{0.3686  $\quad$   (2.0573)}\\ 
    	\midrule
    	\midrule
    	N & \multicolumn{9}{l}{125,995}\\
    	\bottomrule
    \end{tabular}
    }
\end{table}
\noindent
The results show clearly that the shifting method provides closer estimates than mid-point regression everywhere, whereas the latter statistically produces different results from the ``directly observed'' value.\footnote{
For more discussion on the empirical example and for other model setups, see \href{https://regulyagoston.github.io/papers/MSV\_online\_supplement.pdf}{online supplement}, Section 8.} Observe that with the shifting method, $\varepsilon$-differential values are larger but still much lower than the commonly used value of 1 in the literature. The differential privacy method, while using a larger value for $\varepsilon$-differential than any of the realized $\varepsilon$-differential value for the discretization method (which means less data protection), produces a worse and imprecise measure. This is not surprising, as it is documented in the differential privacy literature that these methods comes with cost of accuracy (see, e.g., \citealp{bowen2020comparative}, \citealp{cai2021costofdp}, or \citealp{bi2023distribution}).
%

\section{Conclusion}
\label{sec:conclusion}

This paper deals with linear models using sensitive variables. We propose to use the discretization process to protect data privacy or increase response rates in surveys. This results in discretized (also called interval-censored) variables, which makes econometric modeling difficult as the conditional expectation cannot be point identified in general.
\par
We propose the \textit{split sampling} method that introduces multiple discretization schemes; thus, instead of using one set of intervals, the sensitive variable is discretized through multiple versions. We combine these discretized realizations in a way that, under some mild conditions, they converge in distribution to the original (unknown) variable. We introduce the shifting method, an easy to implement algorithm for split sampling that preserves data privacy while also satisfying the conditions for convergence in distribution. With the help of the shifting method, we can point identify parameters of interest in linear models. The necessary point identification assumptions depend on where the discretization happens: i) the sensitive variable is an explanatory variable; ii) the sensitive variable is the outcome; or iii) discretization happens on both sides. We examine each case and derive the appropriate OLS estimators in a multivariate regression setup. We show the asymptotic properties of these estimators and discuss extensions to nonlinear models and panel data. We provide some Monte Carlo evidence to show that our methods have superior finite sample properties compared to the ``usual'' ones. Finally, we apply our method to estimate the Australian gender wage gap. We achieve not only smaller $\varepsilon$-differential values that show better data protection properties, but consistent parameter estimates even when the number of intervals is small.

\bibliographystyle{elsarticle-harv}
\bibliography{msv_references.bib}

\end{document}